\documentclass[12pt]{article}

\usepackage{amssymb}
\usepackage{amsmath}
\usepackage{amsthm}
\usepackage{url}  % Formatting web addresses  
\usepackage{graphicx}

\usepackage{geometry}
\usepackage{setspace}

\usepackage{fancyhdr}
\newtheorem{theorem}{Theorem} 
\newtheorem{corollary}{Corollary} 

\usepackage{lipsum}

\begin{document}

\title{The Effect of Signaling Latencies and Node \\Refractory States on the Dynamics of Networks}
\author{Gabriel A. Silva (gsilva@ucsd.edu)\\Departments of Bioengineering and Neurosciences\\Center for Engineered Natural Intelligence\\University of California San Diego \\\\ {\normalsize [Note: Accepted version in press in Neural Computation]}}
%\date{\today}
\maketitle

 %\blfootnote{Email: gsilva@ucsd.edu \\Note: Accepted version in press in Neural Computation}

%\doublespacing
\begin{abstract} %\normalsize
 We describe the construction and theoretical analysis of a framework derived from canonical neurophysiological principles that model the competing dynamics of incident signals into nodes along directed edges in a network. The framework describes the dynamics between the offset in the latencies of propagating signals, which reflect the geometry of the edges and conduction velocities, and the internal refractory dynamics and processing times of the downstream node receiving the signals. This framework naturally extends to the construction of a perceptron model that takes into account such dynamic geometric considerations. We first describe the model in detail, culminating with the model of a geometric dynamic perceptron. We then derive upper and lower bounds for a notion of optimal efficient signaling between vertex pairs based on the structure of the framework. Efficient signaling in the context of the framework we develop here means that there needs to be a temporal match between the arrival time of the signals relative to how quickly nodes can internally process signals. These bounds reflect numerical constraints on the compensation of the timing of signaling events of upstream nodes attempting to activate downstream nodes they connect into that preserve this notion of efficiency. When a mismatch between signal arrival times and the internal states of activated nodes occurs, it can cause a break down in the signaling dynamics of the network. In contrast to essentially all the current state of the art in machine learning, this work provides a theoretical foundation for machine learning and intelligence architectures based on the timing of node activations and their abilities to respond, rather than necessary changes in synaptic weights. At the same time, the theoretical ideas we developed are guiding the discovery of experimentally testable new structure-function principles in the biological brain.  
\end{abstract}

\section{Introduction}
We describe the construction and theoretical analysis of a framework that models how local interactions among connected node pairs can be used to compute the global dynamics of a class of spatial-temporal network. By analyzing the local dynamics of incident signals into nodes, and how signals compete to activate downstream nodes, we were able to prove a number of properties associated with the dynamics of the network. This framework was derived from a theoretical abstraction of the canonical neurophysiological principles of spatial and temporal summation in biological neurons \cite{Magee2000,Kandel2012}. We considered spatial-temporal networks in the sense that we analyzed signals along (geometric) directed edges in a metric space that give rise to temporal latencies or delays of signals. The metric space is not intended to conform to any specific coordinate system. The only requirement is that it be able to support a spatial embedding that preserves geometric distances between edges, because it is the geometry of the edges that gives rise to the signaling latencies given a conduction velocity (signaling speed).  When combined with a refractory state for each node, a central concept in this model, the interaction between the two has a dominating effect on the global network dynamics. Following a successful activation, every node experiences a subsequent period, or state, of refractoriness, during which it is incapable of responding to subsequent inputs. The framework we developed reflects local processes from which global behaviors emerge. 

At its core, our model takes into account how the timing of different signals influences their competition to ‘activate’ target nodes they connect into, given that the internal (refractory) state of the node they are competing for may or may not allow that node to be activated. The refractory period is a value that reflects the internal state of the node in response to its activation by an upstream winning node (or set of nodes). In contrast to most other work, we do not assume anything about the internal model that produces this refractory state. For example, this could include an internal processing time during which the node is making a decision about how to react, or an internal re-set period after the node has reacted before it is in a state capable of responding to new inputs. A refractory period is a reasonable assumption for any  physically constructible network, since infinitely fast processing times are not realizable. If the time scale of such a refractory state is much shorter than the time it takes for propagating signals to reach their vertices, due either to fast signaling speeds and/or short edges, it can have a significant effect on network dynamics and function. The inverse will also effect the dynamics, where the refractory period is much longer than the time it takes for signals to arrive at target nodes. These concepts are foundational to the analysis we present in this work. 

In order to intuitively capture the relationship between signal flows on edges and the refractory states of the participating nodes, we defined a refraction ratio for each node. For every pair of connected nodes in the network, we define this as the ratio between the amount of time left in the recovery from the refractory period for the node receiving a signal, relative to the amount of time before the next signaling event arrives at that node from the upstream node that connects into it. The refraction ratio captures the relationship between the speed and temporal delays of signaling or information flow, which is bounded by the spatial geometry of the edges, and the refractory state of the node. We use the term activation here to imply an appropriate reaction of that node to the signal it has received. In the case of biological neurons for example, it would be the generation of an action potential that propagates down the axon and axonal arborizations to synaptic terminals. In neurons, the membrane refractory period due to the biophysics of the membrane ensures the directional nature of the traveling action potential wave along the axon, and sets the cell's maximal firing frequency.  Analogously in our framework, it is the generation of discrete signaling events that directionally travel between edges that connect the nodes in a network. As we and others have previously shown, the interplay between temporal latencies of propagating discrete signaling events, relative to the internal dynamics of the individual nodes, can have profound effects on the dynamics of a network \cite{buibas,Manor:1991tu,Holme:2013p,George:2013b}. In a spatial-temporal network, as we define it in this work, temporal latencies need not be directly assigned but result from the relationship between signaling speeds (conduction velocities) and the geometry of the edges on the network (i.e. edge path lengths). 

An analysis of the framework led to two main results. First, a series of systematic extensions of the basic construction resulted in a spatial-temporal version of a perceptron. In the classical perceptron, contributing weight summations from upstream nodes into an activation function determine whether that node produces an output signal \cite{ros,Minsky:1967ml}. In our case, we considered how temporal latencies produce offsets in the timing of the summations of incoming discrete signaling events. Signaling latencies affect when and how the threshold of the perceptron's activation function is reached, as a function of a running summation in time. This produces a much richer dynamical repertoire relative to the classical perceptron model. As part of this model, we defined a decaying memory function such that the maximum value of the weight for a given connection occurs at the time of arrival of the signal. Subsequent time steps result in progressive decreasing (i.e. decaying) contributions from the maximum weight value to the summation function.

Secondly, a theoretical result from the analysis of the refraction ratio was a set of bounds  that allowed us to formally define a notion of efficient signaling within the context of the framework. We show that an optimal ratio is one where the timing of information propagation between connected nodes does not exceed the internal dynamic time scale of the nodes. In other words, it represents a balance between how fast signals propagate through the network relative to the time needed for each node to process incoming signals. Efficient signaling in the context of the framework we develop here means that there needs to be a temporal match between the arrival time of the signals relative to how quickly nodes can internally process signals. When a mismatch between these two considerations occurs, it can cause a break down in the signaling dynamics of the network. We have previously shown in numerical experiments that a qualitative imbalance of this phenomenon can result in the breakdown of a network to sustain recurrent activity \cite{buibas}. Intuitively, efficient signaling as we develop it in this work implies that there is necessarily a temporal match between the amount of time it takes signals to travel between connected node pairs in a network relative to how quickly the nodes receiving signals can internally process them. A mismatch in time between these two considerations results in a break down in the signaling dynamics of the network. It is this notion of efficient signaling that we formalize here.

The focus of our exposition here is the development and theoretical analysis of the framework, and not a numerical or computational investigation. We also intentionally do not show any applications of the work, which are beyond the intent of this paper. On-going work is using the framework to develop a new machine learning architecture. This architecture is fundamentally different from existing artificial neural network models in how it learns and encodes information and data. We are also using the refraction ratio in the neurophysiological analyses of structure-function dynamics in biological neurons and networks. Using high resolution morphological reconstructions of axon shapes, we recently showed that the refraction ratio reflects a design principle that biological neurons optimize. It reflects a balance between the wiring lengths of axons (material cellular costs) versus signaling speeds (temporal costs) \cite{Puppo2018}.

The paper is organized as follows: In section \ref{sec:model} we introduce the basic construction of the competitive refractory framework and provide an  analysis of its dynamics. Following a set of preliminary definitions (sections \ref{sec:prelims} and \ref{sec:id}), we introduce the refraction ratio in section \ref{sec:refracratio}. We then show how we can use this ratio to compute the set and order of winning nodes in parallel across a network, essentially providing an algorithm to compute global behaviors from local dynamics at the scale of individual node pair interactions (section \ref{sec:analysis}). Section \ref{sec:ext} outlines a number of systematic extensions of the basic construction. In section \ref{sec:prob} we discuss a simple probabilistic extension of the deterministic version of the framework. Section \ref{sec:inh} introduces inhibitory inputs into a node (in contrast to excitatory inputs). In section \ref{sec:intproc} we discuss the role and contribution of internal node processing times to subsequent signaling latencies on edges. Section \ref{sec:per} concludes with the development of fractional node contributions summating in time in order to reach an activation threshold that activates a downstream node. This is our model of a geometric dynamic perceptron. Section \ref{sec:vis} introduces two distinct versions of graphics derived from Feynman diagrams intended to provide a visualization tool for the dynamics. In section \ref{sec:optflow} we introduce a notion of optimal efficient signaling within the context of the framework. We define a set of bounds the refraction ratio must conform to. Section \ref{sec:con} provides some concluding comments. 

\section{Competitive refractory framework}
We considered the geometrical construction of a spatial-temporal network in the following sense. We assume that signals or discrete information events propagate between nodes along directed edges at a finite speed or conduction velocity, resulting in a temporal delay or latency of the signal arriving at the downstream node. Imposing the existence of signaling latencies implies a network that can be mapped to a geometric construction, where individual nodes could be assigned a spatial position in space in $\mathbb{R}^3$ for an ordered triplet $\bar{x}_i = (x_1,x_2,x_3)$ for each vertex $v_i$ for all vertices $i = 1 \ldots N$. Where $N$ is the number of nodes and therefore the size of the network. We will use the terms vertices and nodes interchangeably throughout the paper. Formally, vertices belong to a graph that models a particular network, with the nodes belonging to the network.  Directed edges connecting node pairs could have a convoluted path, i.e. a Jordan arc. There is no restriction that edges have to be spatially minimizing straight line edges (Fig. 1). A signaling latency $\tau_{ij}$ expresses the ratio between the distance traveled on the edge, $d_{ij} = |e_{ij}|$, relative to the speed of the propagating signal $s_{ij}$, between a vertex $v_i$ that connects into a vertex $v_j$. For all pairs of connected vertices $v_iv_j$ $\tau_{ij} = d_{ij} /s_{ij}$. While one does not have to explicitly consider $d_{ij}$ and $s_{ij}$, the existence of signaling latencies can always be mapped to these variables. This is analogous to the conduction velocity of action potentials traveling down the convoluted axon and axonal arborizations of a biological neuron. (\emph{c.f.} Fig. \ref{fig:neteg}B).\label{sec:model}

\begin{figure}
\begin{center}
\includegraphics[width=6.5in]{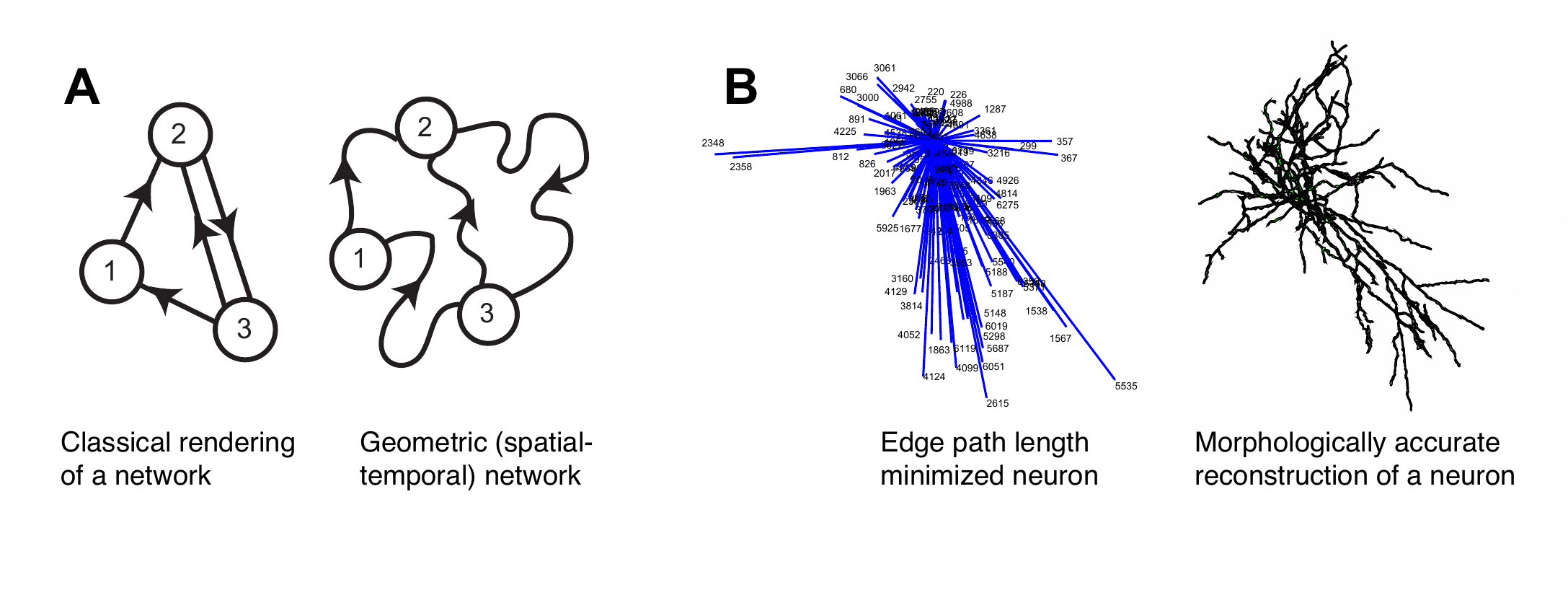} %\doublespacing
\caption{\textbf{A.} Spatial-temporal geometric graph models of a network as we define them here are graphs who's vertices can have a physical position in space (or the plane), with directed edges that have a physical (geometric) representation in space (or the plane).  Given a finite speed for discreete signals propagating down an edge, this produces signaling latencies or delays that when combined with a node (vertex) refractory state results in the dynamics of the framework. In constrast, in a classical representation of a graph the only important consideration is the connectivity of the vertices (i.e. adjacency matrix). There is no geometric signficance to how the graph is drawn. See text. \textbf{B.} An example of a real world spatial-temporal geometric network. Biological neurons can be modeled as tree graphs with the initial node at the cell body (soma) where the axon begins. The axon and its arborizations have convoluted morphologies (edge paths) and display discrete signals that propagate at finite speeds (action potentials). They are physically not edge path length minimized. (Panel (b) adapted from \cite{Puppo2018}.) }
\label{fig:neteg}
\end{center}
\end{figure}

\subsection{Preliminaries}
The set of all edges in the graph $G=(V,E)$ is given by $E=\{e_{ij}\}$, where $e_{ij}$ denotes the directed edge from vertex $v_i$ to vertex $v_j$. We define the subgraph $H_j$ as the (inverted) tree graph that consists of all vertices $v_i$ with directed edges into $v_j$. We write $H_j(v_{i})$ to represent the set of all vertices $v_i$ in $H_j$ and $H_j[v_{i}]$ to refer to a specific $v_{i} \in H_j(v_{i})$. We assume there exist discrete signaling events traveling at a finite speed $s_{ij}$ on the edge $e_{ij}$. The signaling speed $s_{ij}$ from $v_i$ to $v_j$ is bounded such that $0 < s_{ij} < \infty$, i.e. it must be finite. We introduce the notation $H_j[v_{i}] \leadsto v_j$ to mean a vertex $v_{i} \in H_j(v_i)$ that causally leads to the activation of $v_j$. In other words, this notation indicates the `winning' vertex who's signal manages to activate $v_j$ (Fig. \ref{fig:neteg2}). \label{sec:prelims}

We then define an absolute refractory period for the vertex $v_j$ by $R_j$. This reflects the internal dynamics of $v_j$ once a signaling event activates it. For example, the amount of time the internal dynamics of $v_j$ requires to make a decision about an output in response to being activated, or some reset period during which it cannot respond to subsequent arriving input signals. We place no restrictions on the internal processes that contribute toward $R_j$. We assume that we do not know and cannot observe the internal model of $v_j$, which could be quite complex, e.g. a family of temporal differential equations that produce its refractory period.  But we do assume we can observe or measure how long $R_j$ is. We also assume that $R_j > 0$, i.e. there cannot exist an infinitely fast or instantaneous recovery time, even though it can be arbitrarily short. This is a reasonable assumption for any physically constructible network. 

\begin{figure}
\begin{center}
\includegraphics[width=6in]{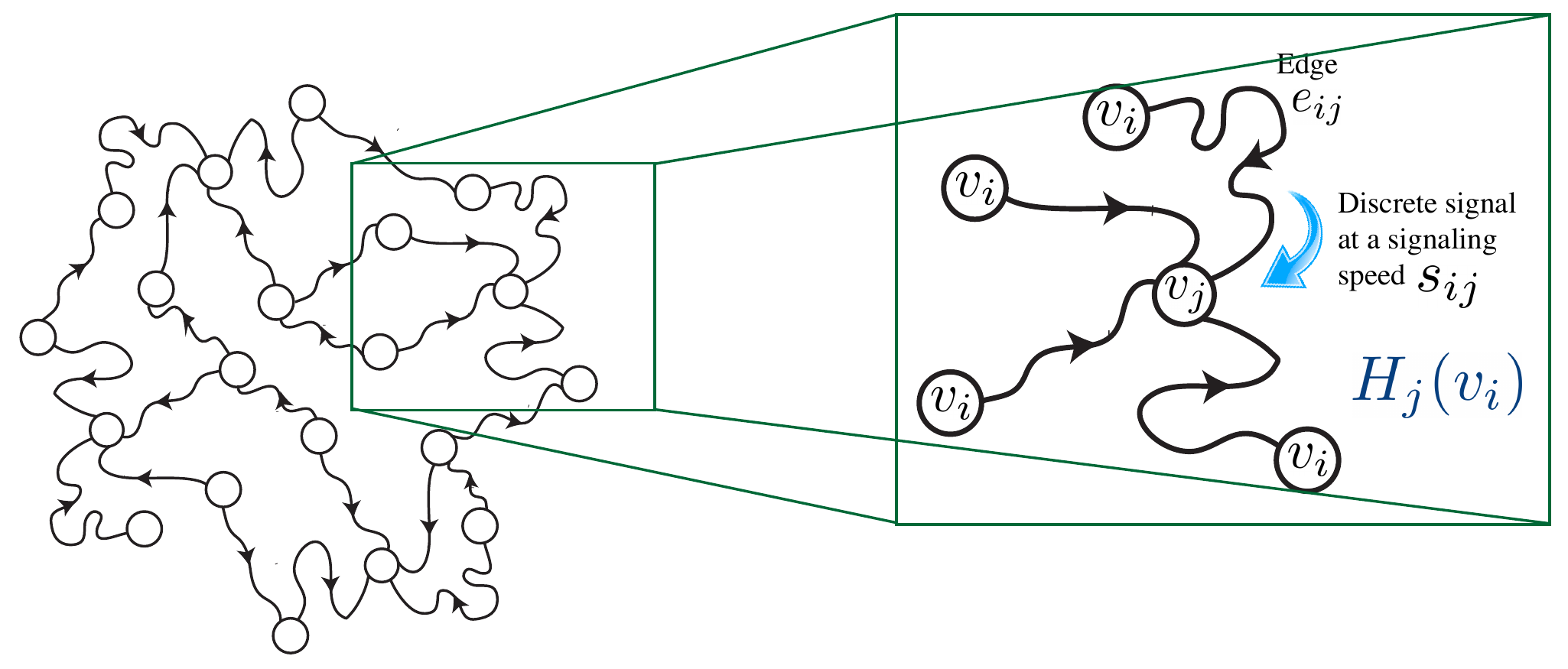} %\doublespacing
\caption{The subgraph $H_j(v_i)$ in relation to an illustrative arbritrary network. Each vertex $v \in V$ has its own $H_j$ subgraph. Note the relative usage of the indices $i$ and $j$. Vertex $v_j$ becomes one of several $v_i$ vertices that in turn connects into other $H_j$ subgraphs. See text for details.}
\label{fig:neteg2}
\end{center}
\end{figure}

Consider a vertex $v_i$ with a directed edge $e_{ij}$ to a vertex $v_j$. For $v_i$ to signal $v_j$, there must be some discrete physical signal representing a flow of information from $v_i$ to $v_j$ over the edge that connects them. This signal must travel at some finite speed $s_{ij}$. $s_{ij}$ could be a constant value for all edges, but this need not be true in the general case. Similarly, if all nodes $v_j$ in a network share the same internal dynamics, then $R_j = R \forall v \in V$. But the framework does not assume this and can accommodate differing node specific values of the refractory period. Once $v_j$ receives a signal from $v_i$ it becomes refractory for a period $R_j$ and will not be able to respond to another incoming signal during this period of time. The temporal nature of $R_j$ implies that as time progresses it shortens and eventually decays to zero, at which time $v_j$ is able to respond to another input.

\subsection{Internal dynamics of $v_j$}  
Let $y_j(\Omega_j,t)$ represent the instantaneous state of vertex $j$ as a function of time and some (possibly unobservable) model or activation function with variable and parameter set $\Omega_j$. The internal state can be interpreted as a binary function determined by $y_j(\Omega_j,t)$. We can define this function at some observation time $T_o$ as \label{sec:id}

\begin{equation} \label{eq:yj}
 y_j(\Omega_j,T_o) =\begin{cases}
    1, & \text{iff $v_j$ can respond to an input}\\
    0, & \text{iff it is refractory to any input}
  \end{cases}
\end{equation} 

Once the winning vertex `activates' $v_j$ it becomes refractory for a period of time $R_j$ during which $y_{j} =  0$, determined by its internal dynamic model. An important point is that if the state of $v_j$ at some arbitrary observation time $T_o$ is $y_j = 0$, it could remain refractory for some time $< R_j$ if it had become refractory prior to $T_o$. This situation is interesting because we have to take into account phase shifts in $\tau_{ij}$ and $R_j$ at $T_o$ in order to understand the patterns of node activations. In other words, which arriving signal results in the activation of $v_j$ is dependent on the amount of the refractory period of $v_j$ remaining relative to the time of arrival of the signal. We call the amount of 'residual' refractory period at $T_o$ the effective refractory period, which we will write as $\bar{R}_j$. This is at the core of our model.  

Algorithmically, we can compute at discrete times in parallel for every vertex in the network, i.e. every $v_j$, which $v_i \in H(i)$ causally activates $v_j$, i.e. $H_j[v_{i}] \leadsto v_j$. We will achieve this by keeping track of the temporal interplay of propagating signaling events on the edges relative to the refractory states of the individual vertices by computing the refraction ratio for all node pairs.

\subsection{The refraction ratio} 
We first establish a simple and intuitive relationship between $R_j$ and $\tau_{ij}$ for a vertex $v_j$. We define the refraction ratio between the refractory period $R_j$ and a signaling latency $\tau_{ij}$ associated with a discrete signaling event coming from a vertex $v_i$ on the edge $e_{ij}$ as 
\begin{equation}\label{eq:R}
\Delta_{ij} = \frac{R_j}{\tau_{ij}} = \frac{R_j \cdot s_{ij}}{d_{ij}}
\end{equation}
Our determination of $H_j[v_{i}] \leadsto v_j$ will emerge from an analysis of this ratio for every vertex connected into $v_j$ in $H_j(v_i)$. \label{sec:refracratio}

Before proceeding further, we point out a number of trivial and unallowable conditions that are necessitated by the physical construction of real world networks and the definitions above. The trivial lower bound occurs as $R_j \rightarrow 0$, $y_j = 1$ at all times. But recall that $R_j > 0$. $R_j =  0$ implies a non-refractory vertex capable of instantaneous recovery to an incoming signal from an upstream vertex, a condition which is not physically realizable. The trivial upper bound occurs as $R_j \rightarrow \infty$, $y_j = 0$ at all times, in which case there would be no information flow or signaling ever. As $\tau_{ij} \rightarrow 0$ $\Delta_{ij}$ becomes undefined, which is equivalent to stating $d_{ij} \rightarrow 0$ since $\tau_{ij} \propto d_{ij}$ for a fixed signaling speed $s_{ij}$. Equivalently, $\tau_{ij} \rightarrow 0$ if $s_{ij} = \infty$. But these conditions are unattainable. $\Delta_{ij}$ is necessarily restricted to finite dynamic signaling and information flow in a network.

 Intuitively, for any $v_i$ into $v_j$ when $y_j = 1$, if every $v_i \in H_j(v_{i})$ initiates a signal at exactly the same moment and $v_j$ is not refractory, the vertex with the shortest edge path will win and activate $v_j$. In other words, assuming a constant signaling speed $s_{ij}$ for all $v_i \in H_j(v_i)$ if we let $D_{ij} :=  \{d_{ij} = |e_{ij}| : i = 1, 2, \dots N\}$ be the set of all edge paths for $H_j(E)$, then $H_j[v_{i}] \leadsto v_j$ will be achieved by $v _{i}(\min_i d_{ij})$ for $d_{ij} \in D_{ij}$. 

Under realistic conditions however, at an arbitrary observation time $T_o$, there is likely to be a temporal offset between when each $v_i$ signals and how far along $v_j$ is in its recovery from its refractory period due to a previous signaling event. Furthermore, each $v_i \in H_j(v_i)$ is statistically independent from every other $v_i$, so that the amount of temporal offset for each $v_i$ vertex signaling $v_j$ will be different. In order to compute these offsets and keep track of the overall dynamics of the network, we need to index two different notions of time. We define $t_i$ to be the moment at which $v_i$ initiates a signal along its edge $e_{ij}$ towards $v_j$. The other is the observation time $T_o$ itself, which is the moment at which we observe or measure the state of $H_j(v_i)$. Or put another way, the moment in time at which we chose to interrogate how far along each signaling event is on its respective edge $e_{ij}$. Such an offset in the progression of discrete signaling events along different edges is the general case, except in the special case when every $v_i \in H_j(v_i)$ signals at the same moment. 

With these considerations in mind, we can take into account temporal offsets by slightly expanding how we define the refractory period and signaling latency. For the refractory period, let $\phi_j$ represent a temporal offset from $R_j$, such that at $T_o$ 
\begin{equation} \label{eq:rbar}
\bar{R}_{j} = R_j - \phi_j \text{ where } 0 \leq \phi_j \leq R_j
\end{equation} \label{eq:effectiveR}
We call $\bar{R}_j$ the effective refractory period. It reflects the amount of time remaining in the recovery from $R_j$ at the observation time $T_o$. When $\phi_j = 0$ it implies $v_j$ became refractory exactly at $T_o$. When $\phi_j = R_j$ it implies that $v_j$ is not refractory and can respond to an input from any $v_i$ at any time. Note how when $\phi_j = R_j$ $v_j$ may have been refractory at some time $t \leq T_o-R_j$, but assures the condition that $\bar{R}_j = 0$ at $T_o$. For values $0 < \phi_j < R_j$, $\bar{R}_j$ is partially recovered. In other words $v_j$ had become refractory before $T_o$ but has not yet fully recovered at $T_o$. Note that unlike $R_j$, which is an actual variable of the dynamical properties of the network, $\phi_j$ and $\bar{R}_j = 0$ are explicitly dependent on the relative observation time $T_o$. One could define them in absolute time as per the discussion in section \ref{sec:abtime} . However, as we argue below, there are some practical advantages and even necessary conditions to consider the state of the network relative to an observation time $T_o$ (\emph{c.f.} section \ref{sec:abtime})

In the most general case the times $t_i$ for all $v_i \in H_j$, i.e. the times at which each vertex initiates a signaling event relative to the observation time $T_o$, would not be expected to be all the same. At $T_o$ a signal from any $v_i$ may be traveling part way along $e_{ij}$ at a speed $s_{ij}$, effectively shortening $\tau_{ij}$. Or it may be delayed in signaling if $v_i$ signals some time after $T_o$, effectively lengthening $\tau_{ij}$. Each vertex $v_i \in H_j(v_i)$ is independent of every other vertex in the subgraph $H_j(v_i)$. And as such we expect different temporal offsets due to how far along a discrete signaling event is in its propagation along their respective edge $e_{ij}$. As such, we need to take into account the distance that a discrete signal initiated by each $v_i$ has traveled along its edge at the observation time $T_o$. To accomplish this, we extend how we consider a signaling latency in the following way. First, we retain $\tau_{ij}$ to represent the absolute temporal delay (latency period) for a signal that travels on the edge $e_{ij}$ for a vertex $v_i \in H_j$ when $v_i$ initiates a signaling event at $t_i$. Note that the initiation of the signaling event could come before, right at, or after the observation time. We then define a temporal offset for $\tau_{ij}$, an effective shortening or lengthening of $\tau_{ij}$ relative to $T_o$ as follows 
\begin{equation} \label{eq:taubar}
\bar{\tau}_{ij} = \tau_{ij} + \delta_{ij} \text{ where, }  \delta_{ij} \in \mathbb{R} 
\end{equation}
$\delta_{ij} > 0$ represents an effective delay or elongation beyond $\tau_{ij}$. In other words, it represents the vertex $v_i$ initiating a signal at some time after $T_o$. Values $-\tau_{ij} < \delta_{ij} < 0$ represent an effective shortening of $\tau_{ij}$. This would be the case when $v_i$ had initiated a signal that was traveling part way along the edge $e_{ij}$ towards $v_j$ prior to $T_o$. Using a consistent notation, we can write this remaining portion of the edge length as $|\bar{e}_{ij}|$. When $\delta_{ij} = 0$ it implies that $v_i$ signals exactly at the moment the network is observed. And when $\delta_{ij} = -\tau_{ij}$ it implies that the signal arrives at $v_j$ at the moment the network is observed. Values of $\delta_{ij} < \tau_{ij}$, which result in $\bar{\tau}_{ij} < 0$, represent a signal arriving at $v_j$ \emph{prior} to the observation time $T_o$.

We can now extend equation \ref{eq:R} to reflect the effective refractory period and effective latency (relative to an observation time $T_o$) as 
\begin{equation} \label{eq:Rovert}
\Lambda_{ij} = \frac{\bar{R}_j}{\bar{\tau}_{ij}}
\end{equation}

The extra degrees of freedom that result from the temporal offsets in the timing of arriving signaling events at $v_j$ for the set $v_i \in H_j(v_i)$ relative to each other,  produce a much larger combinatorial solution space for how $v_j$ may be activated, although we do not  systematically fully investigated the size of this space in this paper. The \emph{global} dynamics of the network results from the \emph{local} statistically independent dynamics of each vertex and its corresponding $H_j(v_i)$. Once $v_j$ is activated it then in turn contributes to the dynamics of the activation of the downstream vertices it connects into. The activation of $v_j$ results in its contribution as one of the $v_i \in H_j(v_i)$ for \emph{each} of the $H_j(v_i)$ it is a part of.   

%We first consider the case where all $v_i$ signal $v_j$ at the same $t_i$ for every $v_i \in H_j$, that is, when $t_i = t_o\forall i \in \Delta^o_{ij} $. 

\subsection{Analysis of the refraction ratio} 
Intuitively,  the `winning' vertex $v_i \in H_j(v_i)$ that successfully achieves the activation of $v_j$, $H_j[v_{i}] \leadsto v_j$, will be the first signaling event that arrives at $v_j$ immediately after $v_j$ has stopped being refractory. This is equivalent to stating that $H_j[v_{i}] \leadsto v_j$ will be achieved by the $v_i$ with the smallest value of $\bar{\tau}_{ij}$ larger than $\bar{R}_j$. This condition guarantees $H_j[v_{i}] \leadsto v_j$. This will be the case for the value of $\bar{\tau}_{ij}$ such that $\tau_{ij} \rightarrow R_j^+$, i.e. approaches $R_j$ from the right, that is, $\tau_{ij}$ is slightly longer than $R_j$. The analysis of $H_j(v_i)$ will therefore necessitate computing the order of arriving signaling events to determine which one meets this condition first. This is effectively an algorithm that computes $H_j[v_{i}] \leadsto v_j$ at $T_o$. \label{sec:analysis}

We begin by defining a well ordered set of refraction ratios for which the condition $\tau_{ij} \rightarrow R_j^+$ is met. First, define the set 
\begin{equation} \label{eq:BigLSet}
\Lambda^o_{ij} := \{\Lambda_{ij} : i = 1, 2, \dots N | \Lambda_{ij} \text{ for } \bar{\tau}_{ij} \rightarrow \bar{R}_j^+\}
\end{equation}
This implies that every $\Lambda_{ij} \in \Lambda^o_{ij}$ satisfies the condition $ \Lambda_{ij} < 1$. This is nothing more than a consequence of how we defined the ratio in equation \ref{eq:Rovert}, with the refractory period in the numerator and the latency in the denominator. Next, we impose an additional structure on $\Lambda^o_{ij}$ by ordering it with the standard $>$ operator. We order the set from the largest refraction ratio to the smallest. Note that although we do not make explicit use of the ordering of $\Lambda_{ij}$ yet, we will in section \ref{sec:per} below. Because $H_j$ is finite, it is a subgraph consisting of a finite number of $v_i$ vertices that connect into $v_j$, $\Lambda^o_{ij}$ will be a finite set. We can then use $\Lambda^o_{ij}$ to compute  $v_i$ $H_j[v_{i}] \leadsto v_j$ using the following theorems. Doing so for all $H_j(v_i)$ in the network in parallel at measurment times $T_o$ allows us to compute the global state and dynamics.

\begin{theorem} Assume $v_j$ has an effective refractory period $\bar{R}_j$ at an observation time $T_o$. If $\phi_j \neq R_j$, then the condition $H_j[v_i] \leadsto v_j$ is given by the refraction ratio $\Lambda_{ij} \in \Lambda^o_{ij}$ that satisfies \label{theo:theo1}
\begin{equation} 
\label{eq:rwithphi}
v_i = \lceil \max(\Lambda_{ij}) \rceil 
\end{equation} 
If, on the other hand, $\phi_j = R_j$ then the condition $H_j[v_i] \leadsto v_j$ is given by 
\begin{equation}
\min({\bar{\tau}_{ij}) \forall v_i \in H_j(v_i})
\end{equation}
\end{theorem}
\begin{proof}
Let the set $\Lambda^o_{ij}$ be defined as in equation \ref{eq:BigLSet}. When $\phi_j \neq R_j$, it guarantees that $\bar{R}_j \neq 0$. This means that the effective refractory period of $v_j$ is either partially recovered from the absolute refractory period at the observation time $T_o$ if $0 < \phi_j < R_j$, or else  $\bar{R_j} = R_j$ when $\phi_j = 0$ (\emph{c.f.} equation \ref{eq:effectiveR}). If $v_j$ is still refractory for some period $\bar{R}_j$, the signal with the smallest $\bar{\tau}_{ij}$ may not be the signal that activates $v_j$ because it may arrive before $\bar{R}_j$ ends. In this case, the signal that arrives first immediately after $\bar{R}_j$ ends will satisfy $H_j[v_i] \leadsto v_j$.  Because of how the refraction ratio and the set $\Lambda_{ij}$ are defined, with $\Lambda_{ij} \in \Lambda^o_{ij} < 1$, (equation \ref{eq:Rovert}), $H_j[v_i] \leadsto v_j$ will be met by the vertex $v_i$ that satisfies $\lceil \max(\Lambda_{ij}) \rceil $. Where the notation $\lceil \cdot \rceil $ indicates the ceiling function, as is typical. To see this, assume some constant value $\bar{R}_j = a > 0$. Let $\bar{\tau}_{(I+1)j} > \bar{\tau}_{Ij}$ for two $\bar{\tau}_{(I+1)j}$ and $\bar{\tau}_{Ij}$ with $\Lambda_{ij} \in \Lambda^o_{ij}$. Then $a/\bar{\tau}_{(I+1)j} > a/\bar{\tau}_{Ij}$. The closest of the two values of $\bar{\tau}_{ij}$ to '$a$' then is $\bar{\tau}_{Ij}$. For a set $\Lambda^o_{ij} = a/\bar{\tau}_{Ij} > a/\bar{\tau}_{(I+1)j} > \cdots a/\bar{\tau}_{ij} \cdots a/\bar{\tau}_{Nj}$, the condition $\bar{\tau}_{ij} \rightarrow \bar{R}_j ^+$ first relative to all $\Lambda_{ij} \in \Lambda^o_{ij}$ is the largest ratio closest to unity, i.e. $\lceil \max(\Lambda_{ij}) \rceil $.

On the other hand, if $\phi_j = R_j$ then $\bar{R}_j = 0 \Rightarrow \Lambda_{ij} =0$ and the 'winning' vertex $H_j[v_i] \leadsto v_j$ will be given by $\min({\bar{\tau}_{ij}) \forall v_i \in H_j(v_i})$, since $\Lambda_{ij} = 0$ only when $\bar{R}_j = 0$ and $\bar{y}_j(\Omega_j,t) = 1$. In other words, if $v_j$ is not refractory at $T_o$ then the first signaling event that reaches it is guaranteed to activate $v_j$. This will be given by the trivial condition of the signal with the smallest effective latency $\bar{\tau}_{ij}$.
\end{proof}

Alternatively, we can equivalently express theorem \ref{theo:theo1} as follows:
\begin{theorem}
Assume $v_j$ has an effective refractory period $\bar{R}_j$ at an observation time $T_o$. Then, for each $v_i \in H_j(v_i)$ with associated $\bar{\tau}_{ij}$, the condition $H_j[v_i] \leadsto v_j$ will be satisfied by \label{theo:theo2}
\begin{equation} 
\label{eq:rwithnophi2b}
 \min \lbrack (\bar{\tau}_{ij} - \bar{R}_j) > 0 \rbrack
\end{equation}
\end{theorem}
\begin{proof}
The necessary condition for $H_j[v_i] \leadsto v_j$ is the vertex with a signal that satisfies $\bar{\tau}_{ij} \rightarrow \bar{R}_j^+$ first, \emph{c.f.} proof of theorem \ref{theo:theo1}. Given pairs of value $(\bar{R}_j,\bar{\tau}_{ij})$ for each $\Lambda_{ij} \in \Lambda^o_{ij}$, and a constant value $\bar{R}_j$ at $T_o$, $max(\Lambda_{ij})$ must also satisfy $\min \lbrack (\bar{\tau}_{ij} - \bar{R}_j) > 0 \rbrack$, given that the ratio of $v_i$ has $\lceil \max(\Lambda_{ij}) \rceil $.  
\end{proof}

Algorithmically, equation \ref{eq:rwithnophi2b} is more efficient to implement because one only needs to compute a difference compared to equation \ref{eq:rwithphi} which necessitates computing a ratio. This becomes significant when computing all $H_j \in G(V,E)$ in parallel. 

\subsection{Dynamics in an absolute versus relative temporal reference frame} \label{sec:abtime} In a strict sense, the criterion for successful activation of a node $v_j$ can be described using a continuous absolute time frame of reference without the need for considering an observation time $T_o$. This can be achieved in a straightforward way. Let $R^{0}_j$ be the start of the refractory period $R_j$ at the moment $v_j$ is activated by $H_j[v_i] \leadsto v_j$. The necessary condition for activation $\bar{\tau}_{ij} \rightarrow \bar{R}_j^+$ can then be re-stated as $R^0_j + R_j  < t_i + \tau_{ij}$. $R^0_j + R_j$ reflects the time at which $v_j$ stops being refractory. If the node becomes refractory at $R^0_j$ and its refractory period lasts for a period $R_j$, then in  it will stop being refractory after $R^0_j + R_j$. The condition for activation will be met when the signal from $v_i$ arrives after this time, namely, after $t_i + \tau_{ij}$. Correspondence with the description above that makes use of a relative framework of observation time can be recovered by the relationship $R^0_j + R_j  - T_o < t_i + \tau_{ij}- T_o$.

While mathematically accurate, there are a number of important functional and practical limitations to this approach that necessitate the consideration of the refraction ratio not from the perspective of absolute continuous time, but from a relative observation or measurement time. A formulation of the model in terms of the absolute times $t_i$ and $R^0_j$ requires that one be able to observe the network from the start of the evolution of the dynamics. In other words, from a dynamically quiescent state. This is because one has to be able to observe the various times $t_i$ and $R^0_j$ for all the nodes in the network as they happen in order to compute the corresponding refraction ratios. A problem arises however if one needs to calculate refraction ratios for a network that has not been observed from the start of its dynamics. In this case, past times $t_i$ and $R^0_j$ from the moment the dynamics starts are unknowable. This is indeed the case for networks that never ‘turn off’ e.g. the brain, or in the context of machine learning frameworks that make use of this model whereby the machine is initiated and allowed to run for a period of time prior to its observation. For example, at the time when data for learning is presented to it and the network computationally interrogated. Similarly, exploring the refraction ratio of neurobiological preparations makes it impossible to know the entire dynamical history of the system. While making appropriate measurements at some $T_o$ may be difficult depending on the system being studied, one is at least guaranteed that computation of the winning vertices and evolution of the dynamics can, in principle, be computed, and avoids a condition that would guarantee an inability to track the system's dynamical history. Namely, needing to know unobserved measurements from the past. 

A related consideration is that even if one were able to observe the dynamics from its inception (i.e. the total history of the system) it could be very computationally expensive to do so. But by referencing a relative observation or measurement time $T_o$, the same calculations are based only on the current (observable) present state of the dynamics, not on its past. This could save a significant amount of memory resources, since the amount of information about the system that needs to be retained is much much smaller.

\section{Extensions beyond the basic construction}
In this section we introduce a number of natural extensions to the basic construction of the framework. In the first three subsections we discuss a probabilistic version of the framework, an approach for dealing with inhibitory inputs, and how the explicit contribution of the internal processing time of nodes affects the timing of subsequent output signals, which in turn affects the refraction ratio. In the last subsection we conclude by considering a version of the framework that accounts for a fractional contribution of summating signals. This in effect represents a geometric dynamic version of the classical perceptron. Note that each of these are not mutually exclusive, and one can consider an implementation of the framework that simultaneously accommodates any or all of them, further adding to the dynamical richness of the model.  \label{sec:ext}

\subsection{Probabalistic extension of the framework}
The construction of the framework as introduced so far is deterministic, in that one individual $v_i$ node is capable of activating $v_j$, $H_j[v_i] \leadsto v_j$. In other words, an activation of $v_j$ by $v_i$,  is guaranteed to produce an output. This is equivalent to stating that the probability of $v_j$ responding to a 'winning' signal from $v_i$ is one. However, we can extend these concepts to add a probabilistic element to the framework. For example, biological neural networks are not deterministically precise, and are often subject to random fluctuations in how one neuron signals another due to thermal dynamical and sub-diffusion considerations associated with molecular events such as neurotransmitter vesicles crossing the synaptic cleft; or the stochastic nature of binding events on the postsynaptic membrane. Such a probabilistic extension would more accurately capture these aspects of the neurophysiology. An open question however, is to what degree such probabilistic considerations improve the functional properties of the model in an engineering sense.  \label{sec:prob}

We assign a probability distribution $P_{ij}$ to the likelihood of activation of $v_j$ for a winning vertex $H_j[v_i] \leadsto v_j$. $v_j(P_{ij})$ indicates the output probability of $v_j$ for some threshold $P_{threshold}$ given that the condition $H_j[v_i] \leadsto v_j$ has occurred:
\begin{equation} \label{eq:vj1}
 v_j(P_{ij}) =\begin{cases}
    1, & \text{if } H_j[v_i] \leadsto v_j \text{ and } v_i(P_{ij}) \geq P_{threshold}\\
    0, & \text{if } H_j[v_i] \leadsto v_j \text{ and } v_i(P_{ij}) < P_{threshold}
  \end{cases}
\end{equation} 
It is important to realize that equation \ref{eq:vj1} does not specify \emph{what} the output from $v_j$ will be given $v_j(P_{ij}) =1$, which could be an actual signal if $v_i$ is excitatory or no signal if it is inhibitory (see section \ref{sec:inh} below), but only that $v_j$ \emph{does} respond in some way to $v_i$. 

\subsection{Inhibitory inputs into $v_j$} %\label{inh}
We can also extend the framework to include inhibitory inputs from $v_i$ in the following way: Given a 'winning' vertex $v_i$ $H_j[v_i] \leadsto v_j$, $v_j$ generates an output signal in turn if the input from $v_i$ was excitatory or does not produce an output but still becomes refractory if the input from $v_i$ was inhibitory. The functional result is the lack of an output from $v_j$ (i.e. no contribution to the activation of \emph{its} downstream vertices), while preventing further activation of $v_j$ until $\bar{R}_j$ ends. This is not the only way one can think of differentiating between excitatory and inhibitory inputs, but it reflects a simple approach. \label{sec:inh}

\subsection{Explicit contribution of the internal processing time of $v_i$} 
In some (most) cases, the activation of a node would not result in an instantaneous output, but require a finite period of processing time associated with the node's internal dynamics prior to the generation and output of a signal. This is a function of the internal dynamic model of the node independent of the network dynamics. The details of such an internal model can be node or node class specific. We do not need to make any such distinction or have any knowledge about what the internal model is. The consequence of an internal processing time on the dynamics of the network through an analysis of our framework is a contribution to the effective latency of the signaling event reaching downstream nodes. The longer the internal processing time the longer the effective latency will be. This is strictly different and independent from, i.e. can evolve in parallel with, the node's effective refractory period. If this internal processing time approaches the time scale of the signaling speed and refractory period then it will affect the dynamics of the network. If however, it is much smaller than both it will have a negligible effect on the dynamics and can be ignored. In the limit as it approaches zero it reflects a (near) instantaneous turn around time between when a signal that activates the node arrives and when that node outputs a signal in turn. \label{sec:intproc}

Before continuing we note a potential source of confusion with regards to the node subscript notation that has to be made clear. An internal processing time is a property of node $v_j$ following its activation. However, it affects the latency of an out-going signal from $v_j$ traveling along the edges it connects to downstream nodes. In the context of signals leaving $v_j$ its subscript in effect changes from $j$ to $i$, because it is now an input into the nodes it connects into. The subscript notation is of course relative to the context of the individual $v_i$,$v_j$ node pairs under consideration. This is the mathematical equivalent of the relative usage of the terms 'presynaptic' and 'postsynaptic' when considering biological neuronal signaling. (A postsynaptic neuron receiving postsynaptic potentials from its presynaptic neurons becomes the presynaptic neuron once its own action potentials reach its synaptic terminals and the cell passes along the `signal' to the neurons it then connects to.) We attempt to be as clear and explicit as possible regarding this distinction in the description that follows. 

Similar to the notation introduced previously, let $t_j$ be the time at which an activated node $v_j$ outputs a signal. Let $t^{in}_j$ be the period of the internal dynamic processing time by $v_j$. This is the time between when a 'winning' signal $H_j[v_i] \leadsto v_j$ arrives at $v_j$, and when $v_j$ actually sends out an output at $t_j$. 

In the description of the framework in section \ref{sec:model} we implicitly assumed that $t^{in}_j = 0$, reflecting an instantaneous output at the moment that $v_j$ is activated. $R_j$ must begin anywhere between the moment of activation of $v_j$ when $\bar{\tau}_{ij}$ for the winning node ends, and $t_j$. But this is strictly a property of the internal dynamics of $v_j$. If $t^{in}_j > 0$, i.e. if there is some period of internal processing time between when $v_j$ is activated and when it sends an output signal at time $t_j$, there will be a lag before $v_j$ initiates an outgoing signal. It affects the timing of when its signaling events reach the vertices it connects to. Recall that for any vertex $v_i$ we denote the time at which it initiates an outgoing signal as $t_i$. This is where the switch in index notation must occur. For an activated vertex $v_j$ with $t^{in}_j > 0$ we write here its signal initiating time as $t'_i$ to indicate that this particular $t_i$ corresponds to the activated node $v_j$.

From a computational perspective we can absorb the effect of $t^{in}_j > 0$ by an appropriate elongation in $\delta_{ij}$. Where again, note that at this point $v_i$ in the $\delta_{ij}$ and $\tau_{ij}$ subscripts in equation \ref{eq:ext} correspond to the original $v_j$ that became activated, thus requiring a switch in index from $i$ to $j$, and the $j$ index now refers to the downstream vertices the original $v_j$ vertex is directionally connected to and signaling. We can re-write equation \ref{eq:taubar} such that for any vertex with $t^{in}_j > 0$ we can express the effective latency $\bar{\tau}_{ij}$ as
\begin{equation} \label{eq:ext}
\bar{\tau}_{ij} = \tau_{ij} + \delta'_{ij} + t^{in}_j = \tau_{ij} + \delta_{ij}
\end{equation} 
where $\delta'_{ij}$ reflects the component of $\delta_{ij}$ that does not account for the extra time due to $t^{in}_j$. $\delta_{ij} = \delta'_{ij}$ modulo $t^{in}_j$. $\tau_{ij}$ is a computed quantity dependent on $s_{ij}$ and $d_{ij}$. But computing $\bar{\tau}_{ij}$ involves taking into account $\delta_{ij}$ such that $\bar{\tau}_{ij} = \tau_{ij} + \delta_{ij}$ \emph{c.f.} equation \ref{eq:taubar}. $\delta_{ij}$ affects when $v_j$ outputs or initiates a signal. The effect of $t^{in}_j > 0$ is to delay by some amount when that occurs at time $t'_i$, i.e. for a period of time $ t^{in}_j$. Because $\delta_{ij} \in \mathbb{R}$ it absorbs $\delta'_{ij}$.

There are three scenarios that fully describe the range of possible effects of an internal signaling delay. 1. a signal $H_j[v_i] \leadsto v_j$ arrives at $v_j$ after which some amount of internal processing time given by $t^{in}_j$ $v_j$ makes a decision to output a signal and does so instantaneously following that decision. 2. a signal $H_j[v_i] \leadsto v_j$ arrives at $v_j$ which makes an instantaneous decision to output a signal but it takes $t^{in}_j$ before the signal from $v_j$ actually goes out. Or 3. a combination of scenarios 1. and 2. whereby some fraction $t^{in}_j/A$ of $t^{in}_j$ represents the time required to make a decision and $t^{in}_j/B = t^{in}_j - (t^{in}_j/A)$ represents the time between when a decision is made and an output signal actually goes out. Note of course that independent of the magnitude of $t^{in}_j$ the ratio $A/[B(A-1)] = 1$. From a practical perspective however, we do not need to distinguish between these three conditions. From the perspective of the framework if $v_j$ produces a signal due to $H_j[v_i] \leadsto v_j$ then $t^{in}_j$ will effectively cause an elongation of $\bar{\tau}_{ij}$.

If for a specific system $t^{in}_j << \tau_{ij} \text{ and } R_j$ it would have no effect on the dynamics. For any real physical system $t^{in}_j$ must always be finite and greater than zero of course, but we can safely ignore its effects if its time scale is much shorter than then the dynamics of the network. The computation of the refraction ratio as given does not change. But the effect on $\Lambda_{ij}$ could affect the global dynamics of the network if $t^{in}_j$ is on the scale of $\tau_{ij} \text{ and } R_j$ .

\subsection{Geometric dynamic perceptrons: summation from fractional contributions of multiple nodes}
In this section we describe the summation of fractional contributions of multiple winning nodes towards the activation of $v_j$. Instead of a single node $v_i$ activating $v_j$, whether deterministically or probabilistically, we now consider a 'running' summation of contributions from a number of $v_i \in H_j(v_i)$ adding up to a threshold value that then activates $v_j$. Upon activation $v_j$ becomes refractory as previously described. Conceptually this represents a competitive refractory model extension of the classical notion of a perceptron, whereby the metric considerations, and therefore signaling latencies, of the model plays a critical role in determining node activation. These perceptrons can be thought of as having a geometric morphology or shape to account for computed latencies on the edges that represent the inputs into $v_j$, similar to biological neurons, e.g. see Fig. \ref{fig:neteg}b. The interplay between the latencies and timing of discrete signaling events on the input edges, and the evolving refractory state of $v_j$, determine the running summation towards threshold of signal contributions from arriving inputs. We assume weights and an activation function as per the standard perceptron model (see below). We will also introduce a decay function that provides a memory or history for previous arriving signals, resulting in diminishing but non-zero contributions towards the summation from inputs that arrived at previous times relative to an observation time $T_o$. This is how we account for a relative contribution of signals given their temporal offsets. Thus, from the perspective of the framework and refraction ratio, the computational prediction being made is not which $v_i \in H_j(v_i)$ will activate $v_j$ after $T_o$, but what \emph{subset} of $H_j(v_i)$ will do so. \label{sec:per}

Consider a $T_o$, and assume that $v_j$ is not refractory. As a function of  time the 'running' summation $\Sigma_r$ from $H_j(v_i)$ must reach a threshold $\Sigma_T$ in order for $v_j$ to activate at some time $t \geq T_o$. Once activated, $v_j$ becomes refractory for a period $R_j$ as usual. The specific contribution from one $v_i \in H_j(v_i)$ will be the value of a weight (synaptic strength), $w_{ij}$. As is typical in a perceptron, we assume a set of weights associated with all incoming connections $W_j = \{w_{ij}\}$. In our model, the maximum value of the contributing weight $w_{ij,max}$ occurs at the time that the signal from $v_i$ arrives at $v_j$ after the end of the relative refractory period for $v_j$, $\bar{R}_j$. This occurs at the time $(\bar{\tau}_{ij} - \bar{R}_j)$ (\emph{c.f.} Theorem 2); explicity, $w_{ij}(\bar{\tau}_{ij} - \bar{R}_j)$ for a time varying weight $w_{ij}(t)$ that decays over time. Beginning at the next time step, we assume the contributing value of the weight starts to decay as a function of time: $w_{ij}(t) < w_{ij,max}$ for times $t > (\bar{\tau}_{ij} - \bar{R}_j)$. In other words, there is a finite memory at future times to the arrival of a given signal, scaled to its weight value, that progressively decays over time to zero. This produces a complex and dynamic interplay between the  latencies of discrete signaling events relative to each other (which in turn encodes an underlying geometry to the perceptrons and the network), the magnitude of the contribution from the respective weights once they do arrive, the kinetics of the decay of the contributing weights as a function of time, and the timing of the refractory state and recovery from the refractory state of $v_j$. Excitatory versus inhibitory weights can be handled as discussed in section \ref{sec:inh} in the sense that inhibitory weights negatively impact $\Sigma_r$. 

Note also that this construction differs from the probabilistic extension of the deterministic version of the framework where one individual $v_i$ is capable of activating $v_j$. For that we assigned a probability distribution to the likelihood of activation of $v_j$ given a winning vertex $H_j[v_i] \leadsto v_j$. Here each $v_i \in H_j(v_i)$ contributes a component to the 'running' summation $\Sigma_r$. One could easily however extend the model further and construct a probabilistic fractional summation version by combining both approaches. This would add even further computational complexity to the total solution space represented by the combinatorial output of the network. However, we do not explicitly take this into account here.

We can again make use of the refraction ratios in the set $\Lambda^o_{ij}$ to compute when $\Sigma_r > \Sigma_T$ with a decaying memory.  Assume a non-refractory vertex $v_j$ at some observation time $T_o$.  We consider the value of the summation of weights into $v_j$ for all contributing $v_i \in H_j(v_i)$. We want to compute what subset of $H_j(v_i)$ will result in a summation that exceeds the threshold value $\Sigma_T$, and at what time $t \geq T_o$ this will occur. Once $v_j$ is activated, it becomes refractory for a period $R_j$. This will involve taking into account the weights $W_j$, values of $\bar{\tau}_{ij}$ for $v_i \in H_j(v_i)$, and the value of $\bar{R}_j$ at $T_o$.

For simplicity in the presentation, we first ignore the decaying memory function, and ask how can we write an expression for $\Sigma_r $. Consider $\Lambda^o_{ij}$ well ordered as in equation \ref{eq:BigLSet}. Due to the ordering of the ratios we can determine the ordering of the arrival of signaling events for each $v_i \in H_j(v_i)$ at $v_j$. This in turn, allows us to determine the contributing weights $w_{ij} \in W_j$ in temporal order that sum to reach $\Sigma_T$. This is the subset of $H_j(v_i)$ that fractionally contributes towards activating $v_j$. We can index this subset $\Lambda_M \subset \Lambda^o_{ij} = \{(\Lambda_{ij})_m \in \Lambda^o_{ij}:m = 1, 2 \ldots M\}$, and then compute 
\begin{equation} \label{eq:w1}
\Sigma_r = \sum_{m=1}^M (w_{ij,max})_m \geq \Sigma_T
\end{equation}
for the corresponding weights. This represents the subset of vertices that participate in activating $v_j$, assuming the summating contributing weights are persistently additive, in other words, do not decay. 

We can, however, take into account a temporal decay of each contributing weight by defining a function $D_i(t)$ that modulates each discrete signaling event following its arrival. When a signal from $v_i$ who's $\Lambda_{ij} \in \Lambda_M$ first arrives at $v_j$, $D_i(t) = 0$, such that there is no decay from $w_{ij,max}$. After some period of time $\xi$ we want $w_{ij} = 0$; with progressively decreasing values in between. We assume $D_i(t)$ is a monotonically increasing function to unity so that the value of $w_{ij}(t)$ progressively decays (eventually to zero) over some period of time. Formally, 
\begin{subequations}
\begin{equation}
D_i(\bar{\tau}_{ij} - \bar{R}_j) = 0
\end{equation}
\begin{equation}
D_i[(\bar{\tau}_{ij} - \bar{R}_j)  + \xi] = 1
\end{equation}
\begin{equation}
\begin{split}
\text{For } &(\bar{\tau}_{ij} - \bar{R}_j)  < t < (\bar{\tau}_{ij} - \bar{R}_j)  + \xi  \text{ if } t_n < t_{m} \\
& \text{then } D(t_n) < D(t_m) \text{, i.e. strictly increasing}
\end{split}
\end{equation}
\label{eq:decayeq}
\end{subequations}
Then, for each individual contributing $v_i$
\begin{equation} %\label{eq:w1}
w_{ij}(t) = w_{ij,max} - w_{ij,max} \cdot D_i(t)
\end{equation}
where as above the domain of $D_i(t)$ is $(\bar{\tau}_{ij} - \bar{R}_j) \leq t \leq (\bar{\tau}_{ij} - \bar{R}_j) + \xi$, and its codomain is $0 \leq D_i(t) \leq 1$. In the same way that we do not restrict the form of the activation function of $v_j$, we intentionally define only a bounded but generalized $D_i(t)$, and do not explicitly restrict its form in any way. The actual rate of change and kinetics of the decay function will of course have a profound impact on the dynamics of contributing weights to $\Sigma_r$, but the exact form of $D_i(t)$ is allowed to vary.

Lastly, we need to consider that $D_i(t)$ for individual contributing $v_i$ will start at different times after the observation time $T_o$. This reflects a time shift to when the decay for each $w_{ij}$ begins relative to $T_o$. This shift determines what is the contribution of each $w_{ij}$ towards $\Sigma_r$ as time progresses. If we think of it in discrete time steps for example, at each successive step the values of $w_{ij}$ from previously arrived signals will be progressively decaying as new signals arrive. Each $w_{ij}$ will contribute different fractional amounts to $\Sigma_r$ up until the point that some $w_{Mj,max}$ due to a last arriving signal from $v_{mj}$ results in $\Sigma_r \geq \Sigma_T$, producing the activation of $v_j$.  Notice how the arrival times of the subset of contributing vertices from $H_j(_i)$ that result in $\Sigma_r \geq \Sigma_T$, taken from $(\Lambda_{ij})_M$, all meet the condition that $(\bar{\tau}_{ij} - \bar{R}_j)  < (\bar{\tau}_{Mj} - \bar{R}_j)$.  To compute the fractional contribution of any given $w_{ij}$ at the moment threshold is reached, in other words, at $t = (\bar{\tau}_{Mj} - \bar{R}_j)$, we need to compute the value of its decay function at $(\bar{\tau}_{Mj} - \bar{R}_j)$ relative to when it started. The decay of the weight from its maximum value will have started after the signaling event arrives. So what we need to compute is how far along the decay is, the value of $D_i(t)$, at the time that reflects the difference from when the signal arrived and when threshold was reached, i..e at $(\bar{\tau}_{Mj} - \bar{R}_j)$. Thus, the individual fractional contribution of each $w_{ij,max}$ can be computed using equation \ref{eq:w1} by evaluating $D_i(t)$ at $$t = (\bar{\tau}_{Mj} - \bar{R}_j) -(\bar{\tau}_{ij} - \bar{R}_j) = (\bar{\tau}_{Mj} - \bar{\tau}_{ij})$$ Finally, the entire summation that includes all weight fractional contributions for the subset of vertices in $H_j(v_i)$ that achieve the activation of $v_j$ that belong to $\Lambda_M$ will be
 \begin{equation} \label{eq:w2}
\Sigma_r = \sum_{m=1}^M (w_{ij,max})_m - (w_{ij,max})_m \cdot D_i(\bar{\tau}_{Mj} - \bar{\tau}_{ij}) \geq \Sigma_T
\end{equation}

Beyond the scope of this paper, we are exploring the construction and use of artificial neural networks constructed from geometric dynamic perceptrons. This work is part of the development of a fundamentally new machine learning architecture that does not necessitate statistical learning methods, and requires no exposure to training data.

\section{Visualizing the dynamics with Feynman-like diagrams} \label{sec:vis}
In this section we introduce two different versions of a diagram that provides a visual representation and summary of the dynamic model and signal summation process. While it takes a bit of explaining to understand how each is constructed, they provide a compact way of visualizing the combination of processes involved, and summarize a lot of information all at once. The inspiration for their structure are the well known Feynman diagrams from particle physics, although this is immaterial for the discussion here. Note that we do not compute or show data for actual networks, since the focus of the current paper is on the mathematics and theory of the model. However, the data computed by the algorithms is capable of generating these plots. Each version of the diagram summarizes the evolving temporal dynamics for a single $H_j(v_i)$ subgraph into a vertex $v_j$ (\emph{c.f.} fig. \ref{fig:neteg2}). There would exist one such diagram (of each version) for all the vertices in the network, i.e. for all $v_j \in G(V,E)$. The entire global evolution of the network dynamics is then summarized by the set of all such diagrams. The intent here however, is to introduce these diagrams as a visualization tool for the evolution of the local dynamics for a specific vertex $v_j$. The difference between the two versions is that in panel A, the entire causal history of $H_j(v_i)$ is preserved for all incoming vertices $v_i \in H_j(v_i)$, while in panel B the diagram shows the combined degree or amount of summation from all the $v_i$ vertices into $v_j$ at any given moment in time. We will discuss each in detail.

\subsection{Visualizing the entire causual history of the dynamics}
Figure \ref{fig:feynmanfig}A provides a visual record of the entire causal history of $H_j(v_i)$ for all incoming vertices $v_i \in H_j(v_i)$ as it evolves in time. In other words, the ordered casual sequence of activation events. Time is represented on the y-axis. So the temporal progression of the dynamics evolves as one moves up the diagram. The initial moment of observation of the network is indicated as the horizontal dotted line at $T_o$. The collection of dots indicating vectors that begin at $T_o$ represents a record of all signals already traveling on their respective edges of the $H_j(v_i)$ subgraph at the moment the network is observed. The x-axis represents the relative spatial positions of where traveling signals are relative to each other and the target vertex $v_j$ in a way we describe below. We will go through each part of the diagram in detail, with the numbering below corresponding to the indicated numbering in the figure. For simplicity and clarity, we only illustrate excitatory inputs, and do not show any inhibitory vectors in the diagram. 

\begin{enumerate}
\item The vertical blue lines represent the state $v_j$ itself. In this particular example as drawn here, at time $T_o$ $v_j$ is partially refractory  (with period of time $\bar{R}_j$ remaining) from a signaling activation at some time prior to $T_o$. This is indicated by the length of the solid blue vector progressing forward in time. Once it stops being refractory the solid blue vector changes to a dotted blue line, indicating $v_j$ is no longer refractory and capable of receiving and responding to subsequent incoming signals. Once signaling events again begin to arrive and summate, the dotted blue line changes to a squiggly blue line to show this. Eventually, once enough summation has occurred, an activation event occurs and $v_j$ again becomes refractory for a period of time, switching back to a solid blue vector. Now though, we observe the full refractory period $R_j$. Of course, the initial state of $v_j$ at $T_o$ for any particular diagram could start with any of these. 

\item Each of the dots at $t = T_o$ correspond to a signaling event that is part away along its edge. The vectors that begin at thesd dots end when they reach $v_j$. The slope of each vector provides a measure of the remaining time it takes the signal to arrive. One can read across to the time axis from the end of the vector once it hits the blue $v_j$ lines. The position along the spatial x-axis is relative in the sense that unlike the time axis, the origins of the vectors at each dot simply indicate how far away each signal is at $T_o$ from the spatial position of $v_j$ relative to each other, i.e. more to the left of $v_j$ correspond to signals that at $T_o$ need to travel farther than signals further to the right. Assuming a constant signaling speed or conduction velocity for all $s_{ij}$, as we have throughout the paper, the remaining time it takes signaling events that are already partially along their edges at the time the graph is observed at $T_o$ will be proportional to the distance on the edge left to travel, $|\bar{e}_{ij}|$. The slope of each vector reflects the remaining edge distance and therefore time required for the signal to reach $v_j$. Note that if the spatial axis was in real units, then the spacing between the origin of each vector would represent the actual physical distance of the signal from $v_j$ and assuming a constant $s_{ij}$ for all signals every vector would be at 45 degrees. This is another valid way of drawing it. Here we chose to keep the spatial axis relative in arbitrary units to keep the diagram more compact and allow the slope to encode the spatial information that produces the remaining time the signals take to reach $v_j$. 

In the illustrative example here, we see that the very first signal arrives while $v_j$ is still refractory and has no effect, the second signal begins a running fractional summation $\Sigma_r$. The next two signals add to $\Sigma_r$. The fourth arriving signal causes $\Sigma_r > \Sigma_T$ and $v_j$ fires, making it refractory again for the period $R_j$. 

\item Any vector that begins at a time later than $T_o$ represents a signal that left its origin from one of the $v_i$ vertices that connect into $v_j$ after the observation time $T_o$. Note how in the example for the vector shown here at position 3 in the diagram, the signal represented by this vector can represent either one of two possible conditions. 1. It originated at a node $v_i$ that did not have a signal partially traversing its edge at $T_o$, or 2. it corresponds to a signal from a $v_i$ that did have a partial signal at $T_o$ and fired again. If condition 2. then the signal must have originated from a vertex $v_i$ that corresponded to one of the first two vectors (signals) to its right at $T_o$, i.e. one of the two vertices closest to the spatial position of $v_j$ at $T_o$, since its total edge length $|{e}_{ij}|$ cannot be shorter than a partial edge distance of $|\bar{e}_{ij}|$ at $T_o$. To be fair, it is not possible to  distinguish these two conditions from the diagram as drawn without a vertex label. However, the algorithms readily keep track of this too if necessary. 

This diagram is able to visualize  the dynamics in a way that summarizes a lot of information that can be quickly gleamed. 

\item Finally, in position 4 in the diagram we show two signals originating from the same $v_i$ but occurring at different times. Note how because we assume that $s_{ij}$ is constant, the slopes for both vectors, which encodes the time it takes for these signals to travel the full $|e_{ij}|$ for a particular $v_i$, will be the same. 
\end{enumerate}

\subsection{Visualizing the dynanmics of the summations $\Sigma_r$}
In Figure \ref{fig:feynmanfig}B there is less information in the sense that  the signaling events that lead to the activation of $v_j$ are compactified into a single vector. But the result is the ability to intuitively visualize the progression of the running fractional summation $\Sigma_r$ that leads to the activation of $v_j$. The time axis on the y-axis remains the same, but the x-axis is no longer the relative spatial position of the signals from vertices $v_i \in H_j(v_i)$ that make up the subgraph. Rather, it is the degree or amount of $\Sigma_r$ itself. The fractional summation culminates at the vertical dotted line indicating the threshold summation value required to trigger an activation or firing of $v_j$ at $\Sigma_T$. 

\begin{enumerate}
\item Vectors in black progress show an accumulation or increase in $\Sigma_r$. Here, they represent the totality of summations producing an increase in $\Sigma_r$. This does not only mean the summation of excitatory inputs adding to $\Sigma_r$, but also counteracting inhibitory inputs that decrease or delay $\Sigma_r$ (depending on how inhibitory inputs are accounted for).  Importantly,  a component of this vector is the decay of weight contributions for both excitatory and inhibitory signals as determined by the decay equation \ref{eq:decayeq}. Black vectors that increase $\Sigma_r$ imply that excitatory summations outpace inhibitory summations that detract from $\Sigma_r$, due to a combination of the number of excitatory inputs and/or frequency of excitatory signals, .

\item Red vectors are the same as black vectors but represent the opposite process: net inhibitory summations that are greater than excitatory summations. This then reduces the value of $\Sigma_r$.

\item Through these processes, there will be a back and forth 'fluctuation' of summating events that push and pull $\Sigma_r$ to and away from $\Sigma_T$. This is the model equivalent of and directly analogous to sub-threshold membrane potential fluctuations in the dendritic trees of biological neurons. In real neurons spatial and temporal summation in dendrites attempt to reach the membrane threshold potential at the initial segment in order to trigger an action potential. As discussed beginning of the paper, this is in effect what our framework models and abstracts out from the neurobiology. It represents the main motivation for this work. In the example shown here at position 3 in the diagram, an inhibitory summation process that is lowering $\Sigma_r$ all of a sudden reverses because the summation of excitatory events overtakes it and starts driving $\Sigma_r$  towards $\Sigma_T$.

\item Here an activation event is illustrated. $\Sigma_r$ reaches $\Sigma_T$, a red vector must follow for at least a period corresponding to $R_j$ due exclusively to the decay function driving down $\Sigma_r$, since $v_j$ will not be able to respond to inputs, excitatory or inhibitory. 

\item Once $R_j$ ends, if the summation conditions are right, $\Sigma_r$ will begin to increase again. But note how the red vector can continue longer than the period dictated by $R_j$ if summation conditions are not favorable. 

\end{enumerate}

It should be clear to the reader how the multiple summating vectors in panel A that trigger an activation at $v_j$ correspond to the equivalent compactified perspective in panel B. The entire mathematical description introduced in the preceding sections of the paper can be visualized in these two ways. In subsequent work, we will explore the use of these Feynman-like diagrams  to study how the local rules that drive the dynamics of the individual $H_j(v_i)$ subgraphs that these diagrams represent produce the global behaviors and dynamics of the overall network. 

\begin{figure}
\begin{center}
\includegraphics[width=6.5in]{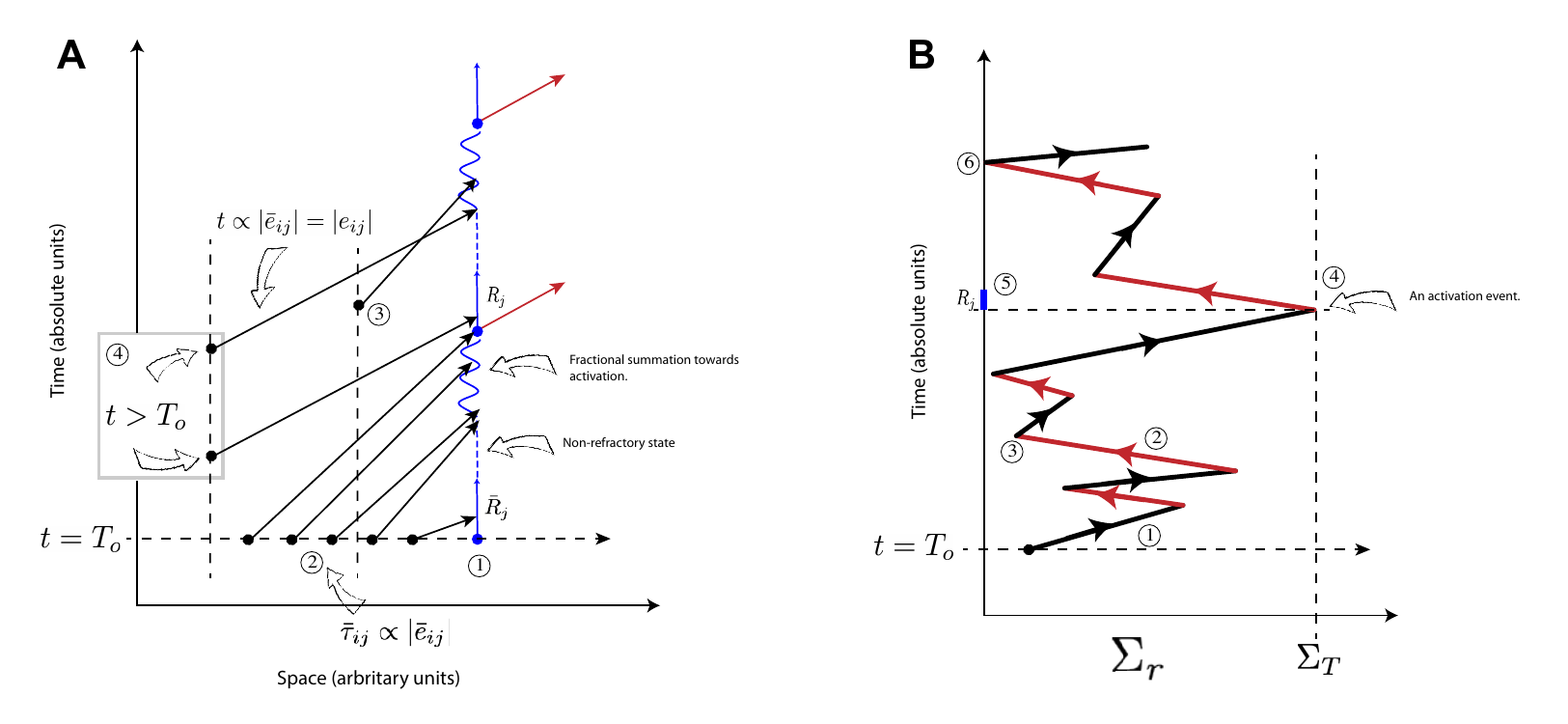} %\doublespacing
\caption{Feynman-like diagrams for visualizing the competitive refractory model and fractional summation of the signal summation process. \textbf{A.} The entire causal history of $H_j(v_i)$ is shown for all incoming vertices $v_i \in H_j(v_i)$. \textbf{B.} A compactified representation of the dynamics which shows the combined amount of summation $\Sigma_r$ from all the $v_i$ vertices into $v_j$ at any given moment in time. See the text for details.}
\label{fig:feynmanfig}
\end{center}
\end{figure}

\section{Efficient signaling  } 
In this last section we prove a notion of optimized efficient signaling in the context of the framework and the refraction ratio. These arguments emerge naturally from a consideration of how the ratio is defined. For convenience, we limit our discussion to the deterministic version of the framework. Given an effective refractory period $\bar{R}_j$ and effective latency $\bar{\tau}_{ij}$ along an  edge $e_{ij}$, $H_j[v_i] \leadsto v_j$, is dependent on the $\lim_{\bar{\tau}_{ij} \rightarrow \bar{R}^+_j}$ for $\Lambda_{ij} = \bar{R}_j/\bar{\tau_{ij}}$. In other words, the first discrete signal that arrives at $v_j$ after it stops being refractory. Intuitively, efficient signaling in the context of the framework means that there should be a temporal match between the arrival time of the signals relative to how quickly nodes can internally process signals. When a mismatch between these two considerations occurs, it can cause a break down in the signaling dynamics of the network. For example, we have previously shown numerically (in computational simulations) that if  signaling speeds $s_{ij}$ are too fast, or equivalently, if the latencies $\bar{\tau}_{ij}$ are too short, compared to the amount of time a node requires to process a signal and recover to a state in which it can respond again to another input, the network will not be able to sustain internal recurrent activity \cite{buibas}. This reflects a form of inefficient signaling in the sense that the network requires constant external driving (energy) to sustain it.

Conversely, if $s_{ij}$ is too slow or the set of $\bar{\tau}_{ij}$ too long then the network will be inefficient in a very different sense. It has the potential for faster dynamic signaling that is not being realized. Time, as a resource, is being wasted in such a network. The following discussion formalizes these concepts. \label{sec:optflow}

%\subsection{Bounded refraction ratio} %%\label{optflow}
We first derive upper and lower bounds for the refraction ratio between connected vertices $v_i$ and $v_j$. We argue that these bounds formalize a notion of optimal efficient signaling that reflect a match between the rate of discrete signal propagation (information flow) between vertices relative to the internal dynamics of the vertices. 
 
\begin{theorem} [Optimized refraction ratio theorem] %\label{theo:opt}
Let $G=(V,E)$ represent a geometric network consisting of subgraphs $H_j(v_i)$. For each $v_i v_j$ vertex pair with a signaling speed $s_{ij}$ between $v_i$ and $v_j$, the optimal refraction ratio $[ \Lambda_{ij} ]_{opt}$ at an observation time $T_o$ is bounded by \label{theo:theo3}
\begin{subequations} \label{eq:opt}
\begin{align}
& [ \Lambda_{ij} ]_{opt} = \lim_{\tau_{ij} \rightarrow R_{j}^+} \Lambda_{ij} \text{ when } \phi_j \text{ and } \delta_{ij} = 0 &&\text{ [Upper bound ]}  
\\ & [ \Lambda_{ij} ]_{opt} \Rightarrow \lim_{\delta_{ij} \rightarrow -\phi_j} \Lambda_{ij} \text{ when } \phi_j = R_j &&\text{ [Lower bound]}
\end{align}
Given these bounds then, an optimized refraction ratio will be such that 
\begin{equation}
[ \Lambda_{ij} ]_{opt} = \lim_{\tau_{ij} \rightarrow \bar{R}^+_j} \frac{\bar{R}_j}{\tau_{ij}} \rightarrow 1
\end{equation}
\end{subequations}
\end{theorem}

\begin{proof}
 The necessary condition for the activation of $v_j$ by $v_{i} \in H_j(v_i)$ is $\bar{\tau}_{ij} > \bar{R}_j$. By equation \ref{eq:rbar} $\bar{R}_j = R_j - \phi_j \text{ where } 0 \leq \phi_j \leq R_j$, which implies that $0 \leq \bar{R}_j \leq R_j$. $\bar{R}_j$ is bounded by its very construction. The absolute lower bound on $\bar{R}_j$ implies that activation of $v_j$ by a $v_{i} \in H_j(v_i)$, $\bar{\tau}_{ij} \rightarrow \bar{R}^+_j$, will be achieved when $\bar{\tau}_{ij} \rightarrow 0^+$, i.e. is just slightly larger than zero, and the absolute upper bound implies that $\bar{\tau}_{ij} \rightarrow R_j^+$. But note how $\bar{\tau}_{ij}$ can always achieve these bounds independent of $\tau_{ij}$ for a given $v_i$ $v_j$ pair at $T_o$ because by equation \ref{eq:taubar} $\delta_{ij} \in \mathbb{R}$, i.e. any vertex $v_i$ can activate $v_j$ independent of the absolute latency $\tau_{ij}$ by delaying the initiation of an output signal at $v_i$ long enough if $\tau_{ij}$ is too short or initiating a signal at $v_i$ prior to $T_o$ if $\tau_{ij}$ is too long. However, by equation \ref{eq:BigLSet} and theorem \ref{theo:theo1}, $\bar{\tau}_{ij}$ need only be slightly larger than $\bar{R}_j$ in order to successfully signal $v_j$, i.e. $\bar{\tau}_{ij} \rightarrow \bar{R}_j^+$. Because $\bar{R}_j$ is naturally bounded by $0 \leq \bar{R}_j \leq R_j$, it follows that the optimal signaling condition will be given by $\tau_{ij} \rightarrow \bar{R}_j^+$ for values of $\delta_{ij}$ not too smaller than zero or not too greater than zero in order to meet the condition that $\tau_{ij} \rightarrow \bar{R}_j^+$ while avoiding compensation by $\delta_{ij}$. In other words, the response dynamic range for any $v_j$ will always be bounded by the limits of $\bar{R}_j$ in the sense that these limits determine the temporal properties of when $v_j$ can actively participate in network signaling and when it cannot. Ultimately, of course, this is a function of $v_j$'s internal dynamics, which in turn determines $R_j$ and $\bar{R}_j$. No value of $\tau_{ij}$ need be much greater than $R_j$ for any $v_i$ with a directed edge $e_{ij}$ into $v_j$. When the condition $\tau_{ij} \rightarrow \bar{R}_j^+$ is met, it ensures that such a $\tau_{ij}$ is guaranteed to be able to operate over the entire response dynamic range of $v_j$, i.e. all values of $\bar{R}_j$. Given this condition for optimized signaling in the context of the refraction ratio then, for the upper bound this optimized boundary condition will occur when $\tau_{ij} \rightarrow R_j^+$ when $\phi_j$ and $\delta_{ij} = 0$ because it represents the upper achievable limit for $\bar{R}_j$ (when $\phi_j = 0$) and forces the optimal condition that $\tau_{ij} \rightarrow R_j^+$ without compensating with $\delta_{ij}$. 

For the lower bound the optimal condition is given by $\bar{\tau}_{ij} \rightarrow 0^+ \Rightarrow \delta_{ij} \rightarrow -\phi_j$ when $\phi_j = R_j$, since when $\phi_j = R_j \Rightarrow \bar{R}_{ij} = 0$. Forcing the condition that $\phi_j = R_j$ implies that $\tau_{ij}$ on its own is capable of meeting the lower bound without compensation by $\delta_{ij}$. We can define an optimized bound as $|\tau_{ij} - R_j| < \epsilon$ for some bounded error $\epsilon$.  If $\bar{\tau}_{ij}$ is too short, either because the path length of $e_{ij}$ is too short or $s_{ij}$ is too fast, this implies that given $R_j$, $\bar{\tau}_{ij} \nrightarrow R_j^+$ if $\delta_{ij} \rightarrow 0$. To achieve the lower bound it would require $\delta_{ij} < 0$ so that $\bar{\tau}_{ij} < \tau_{ij}$. To achieve the upper bound it would require $\delta_{ij} > 0$ so that $\bar{\tau}_{ij} > \tau_{ij}$. If $\bar{\tau}_{ij}$ is too long, either because the path length of $e_{ij}$ is too long or $s_{ij}$ is too slow, this implies that given $R_j$, $\bar{\tau}_{ij} \nrightarrow 0^+$ if $\delta_{ij} \rightarrow -\phi_j$. When these constraints are met, it ensures that $\lim_{\tau_{ij} \rightarrow \bar{R}^+_j} \bar{R}_j/\tau_{ij} \rightarrow 1$. For the lower bound the important condition is that $\bar{\tau}_{ij} \rightarrow 0^+$ when $\delta_{ij} \rightarrow -\phi_j$. It is trivial what $\delta_{ij}$ is, since this condition will always be met when $\delta_{ij} = -\tau_{ij}$. But for the upper bound the important condition is that $\tau_{ij} \rightarrow R_j$ when $\delta_{ij} \rightarrow 0$, which implies that in every case $\bar{\tau}_{ij} \rightarrow \tau_{ij}$.
\end{proof}

 What this implies is that if $v_i$ satisfies the bounding conditions, it will always be within a range where it could efficiently 'win' and activate $v_j$ for any value of $\bar{R}_j$ at an observation time $T_o$ and time of signaling initiation by $v_i$ at $t_i$. By efficient we mean without forcing $\delta_{ij}$ to compensate. The given $v_i$ may not of course always 'win' in activating $v_j$ but it is asured to be as efficient as possible, as efficient as any other node in its signaling of $v_j$ over all values of $\bar{R}_j$. If $R_j$ is the same for all nodes in the network, $R_j = R$ $\forall j \in G(V,E)$, this in effect bounds the dynamic window over which all network dynamics, that is, all temporal information (signaling) processes, should occur on that network. Explicitly, this dynamic window is given by $\bar{R}_j = R_j - \phi_j \text{ for } 0 \leq \phi_j \leq R_j$. So there is no need or reason for any $\tau_{ij}$ to go far beyond this dynamic window in order to satisfy the optimality condition $\bar{\tau}_{ij} \rightarrow \bar{R}_j^+$. This is in essence the intuitive idea of using the bounds derived in theorem \ref{theo:theo3} to define optimized signaling or information flow between node pairs. Note that we must keep the explicit condition that $\phi_j = R_j$ because that forces $\Lambda_{ij} = 0$ only when $\delta_{ij} \rightarrow -\phi^+$. Otherwise, in the general case any value of $\bar{\tau}_{ij} >> 0$ will result in $\Lambda_{ij} = 0$ for any value of $\bar{R}_{j}$.

We can also prove the following corollaries.
\begin{corollary}
Given the conditions for lower and upper bounds in theorem \ref{theo:theo3}, let $[ \delta_{ij} ]_{upper}$ denote the value that $\delta_{ij}$ must take in order to achieve the optimal upper bound condition for some value of $\tau_{ij}$. Similarly, let $[ \delta_{ij} ]_{lower}$ denote the value that $\delta_{ij}$ must take in order to achieve the optimal lower bound condition. In every case, the relationship between $[ \delta_{ij} ]_{upper}$ and $[ \delta_{ij} ]_{lower}$ is given by
\begin{equation*}
[ \delta_{ij} ]_{lower} = [ \delta_{ij} ]_{upper} - R_j \\
\end{equation*}
\end{corollary}

\begin{proof}
The condition for the upper bound is $\tau_{ij} \rightarrow R_j^+$ for $\delta_{ij} = 0$ (and $\phi_j =0$). Since $\bar{\tau}_{ij} = \tau_{ij} + \delta_{ij}$, if $\tau_{ij} \neq R_j$ then $\bar{\tau}_{ij} \nrightarrow R_j^+$. For a given $\tau_{ij}$ and for a known or measurable $R_j$, $[ \delta_{ij} ]_{upper} = -(\tau_{ij} - R_j)$, since this is what the value of $\delta_{ij}$ would have to be in order to achieve the optimal condition. For the lower bound $\bar{\tau}_{ij} \rightarrow 0^+$ when $\delta_{ij} \rightarrow -\phi_j$ given that $\phi_j = R_j$. Thus, if $\tau_{ij} -\phi_j = \tau_{ij} -R_j > 0$, or more correctly if $\bar{\tau}_{ij} \nrightarrow 0^+$, it implies that $[ \delta_{ij} ]_{lower} = -\tau_{ij}$ would be needed to meet the lower bound optimality condition. The difference between $[ \delta_{ij} ]_{lower}$ and $[ \delta_{ij} ]_{upper}$ is therefore $- \tau_{ij} - [- (\tau_{ij} - R_j)] = - R_j$. Thus, $[ \delta_{ij} ]_{lower} = [ \delta_{ij} ]_{upper} - R_j$. 
\end{proof}

\begin{corollary}
If a signal $v_i \in H_j(v_i)$ characterized by a latency $\tau_{ij}$ on the edge $e_{ij}$ is able to achieve either the optimal upper bound or optimal lower bound as per theorem \ref{theo:theo3}, then it is guaranteed to be able to achieve the other optimal bound.
\end{corollary}

\begin{proof}
Asking if a signal capable of achieving the upper bound can also achieve the lower bound is equivalent to asking if $\bar{\tau}_{ij} \rightarrow 0^+$ when $\tau_{ij} = R_j$. But the condition for the lower bound is $\bar{\tau}_{ij} \rightarrow 0^+$ when $\delta_{ij} = -\phi_j$. Substituting for these explicit variables we arrive at 
\begin{align*}
\bar{\tau}_{ij} &= \tau_{ij} + \delta_{ij}\\
&= R_j - \phi_j
\end{align*}
but since $\phi_j = R_j$ for the lower bound, it implies that $\bar{\tau}_{ij}  = 0$, or more appropriately, $\bar{\tau}_{ij}  \rightarrow 0^+$. 

Asking if a signal that satisfies the optimality condition for the lower bound can also achieve the upper bound is equivalent to asking if $\bar{\tau}_{ij}  \rightarrow R_j^+$ when $\delta_{ij} \rightarrow -\phi_j$. Similarly, 
\begin{align*}
\bar{\tau}_{ij} &= \tau_{ij} + \delta_{ij}\\
0 &= \tau_{ij} - R_j\\
\tau_{ij} &= R_j
\end{align*}
which implies that the optimality condition for upper bound is satisfied. 
\end{proof}

A strict definition of an optimally efficient network then follows: In every case, as a function of the effective refractory period $\bar{R}_j$ and effective delay time $\bar{\tau}_{ij}$ along the edge $e_{ij}$, the condition for the winning vertex $v_i$ that achieves activation of $v_j$, i.e. $H_j[v_i] \leadsto v_j$, is dependent on the $\lim_{\bar{\tau}_{ij} \rightarrow \bar{R}^+_j} \forall v_i \in G(V,E)$. When this condition is satisfied for all edges $e_{ij} \in E =\{e_{ij}\}$, for all $v_i$$v_j$ node pairs by the upper and lower bound definitions for $[ \Lambda_{ij} ]_{opt}$ in theorem 3 (equation \ref{eq:opt}) such that $\Lambda_{ij} =  \bar{R}_j / \bar{\tau}_{ij} \rightarrow 1$, the network is optimally efficient, i.e. $[ \Lambda_{ij} ]_{opt} \forall v_i \in G(V,E)$. This is equivalent to requiring the condition $|\tau_{ij} - R_j| < \epsilon$ $\forall v_i \in G(V,E)$ for some arbitrarily small value of $\epsilon$.

A consequence of this is that for a network to meet this strict definition it must exhibit a lattice structure. Given constant values of $|t_{j}|$ and $R_j$ and $s_{ij}$ $\forall e_{ij} \in E$, optimized efficient signaling at the network scale is only achievable when $d_{ij} = |e_{ij}| = C$ $\forall e_{ij} \in E$, where $C$ is a constant. Geometrically, this implies that the network must exhibit a lattice structure. To see this, note that if we assume optimal signaling between any arbitrarily chosen but specific pair of vertices $v_I,v_J$ such that $\Lambda_{IJ} \rightarrow 1^+$, then
\begin{equation*}
\Lambda_{ij} = \frac{R_J - \phi_J}{\tau_{IJ} - \delta_{IJ}} = \frac{R_J - \phi_J}{\frac{d_{IJ}}{s_{IJ}} + \delta_{IJ}} \rightarrow 1^+
\end{equation*}
For the lower bound $\phi_i$ and $\delta_{ij} = 0$, so the condition for optimal signaling will be $R_j \cdot (s_{ij}/d_{ij}) \rightarrow 1^+$. This means that for any $v_iv_j$ pair such that $d_{ij} \neq d_{IJ}$ $\Lambda_{ij} \nrightarrow 1^+$.

For the upper bound $\phi_j =\delta_{ij}$, so 
\begin{equation*}
\Lambda_{ij} =  \frac{R_j - \delta_{ij}}{\frac{d_{ij}}{s_{ij}} + \delta_{ij}} \rightarrow 1^+
\end{equation*}
We know that $[ \delta_{ij} ]_{lower} = [ \delta_{ij} ]_{upper} - R_j := K$. Then
\begin{equation*}
\Lambda_{ij} =  \frac{R_j - K}{\frac{d_{ij}}{s_{ij}} + K} \rightarrow 1^+
\end{equation*}
Then evaluating again for the lower bound given the substitution yields the same result as above, such that $\Lambda_{ij} \nrightarrow 1^+$. Of course, while this provides a strict criteria of optimized network signaling by evaluating the optimality of signaling across all node pairs independently, it is not necessarily equivalent to optimized network dynamics. In other words, a pure lattice structure may not be (and in most cases will not be) the best geometric structure to optimize the dynamical flow of information across a network in support of some function or behavior or learning the network is intended to do. This condition is strictly mathematical as a consequence of theorem \ref{theo:theo3}.

\section{Concluding comments}
In this paper we presented an intuitively simple framework that describes the competing dynamics of signaling and information flows in spatial-temporal (geometric) refactory networks derived from foundational principles of biological cellular neural signaling. These results provide a systematic explanation for how the interplay between strictly local geometric and temporal process at the scale of individual interacting nodes ultimately give rise to the global behavior of the network. At the core of the model is the assumption that the response of each vertex is causally independent from whatever all the other vertices in the network are doing, even though collectively they contribute to the global dynamics. The internal (to each vertex) determination of the state $y_j$ is not dependent on any 'average' metric of the state of the network as a whole, or on statistical probability densities associated with the frequency of occurrence of events such as in Markovian processes.  This allows us to model and simulate network dynamics by keeping track of the latencies associated with signals on the edges, without needing to account for or model any number of dynamic variables necessary to describe the internal models of the individual nodes that we may have little to no information about. Or necessitating the computation of statistical variables from which network dynamics is then inferred, which generally preclude real time analysis due to the need of observing or measuring sufficient data first. This is in direct contrast to most other network analysis approaches that attempt to measure or capture global properties. \label{sec:con}

The class of networks we define and investigate here can be thought of as belonging to a broader class of networks refered to as spatial-temporal and diffusion networks. These networks are of relevance to many natural, engineering, and technological systems, because the conceptual model they provide naturally maps to many real world physical systems (\cite{baagelj2014,cencetti2018,George:2013b,Gomez2011,gong2018,Holme:2013p,shan2018}. Examples include transportation systems, communication networks, social networks, the spread of infections diseases (including computer viruses and malware), physical-chemical systems such as the interactions of particles in solution and diffusion, and -omics type biological networks of molecular and genetic interactions. However, much of the literature on these classes of networks  has focused on descriptive theory, definitions of concepts and notation, and various functional metrics. There are comparatively fewer applications, algorithms, and testable predictions. There is an almost exploratory quality to the existing literature. In particular, there is no work we are aware of that develops a theoretical foundation that takes advantage of canonical neurobiological processes to arrive at a set of practical algorithms, like we do here, with applicability to both neuroscience and machine learning. We have previously argued that theoretical neuroscience, while challenging, offers an opportunity for a deeper understanding than purely numerical or simulation based computational neuroscience models \cite{Silva2011}.

As an example of a real world application of the ideas and framework we discuss here, of particular interest to us is the signaling dynamics responsible for computation in biological neural networks across scales of organization (e.g. dendritic trees, neural circuits, or the interaction between different brain regions) constrained by the underlying structure. In other words, the structure-function relationships responsible for neural dynamics. We recently showed that Basket cell neurons in the cortex are putatively designed to optimize for the refraction ratio. They seem to be designed to optimize the ratio between the refractory period of the membrane, and action potential latencies along the axon between the initial segment (where action potentials begin) and the synaptic terminals (where they end) \cite{Puppo2018}. Dynamic signaling on branching axons is critical for rapid and efficient communication between neurons in the brain. Efficient signaling in axon arbors depends on a trade-off between the time it takes action potentials to reach synaptic terminals (temporal cost) and the amount of cellular material associated with the wiring path length of the neuron’s morphology (material cost). However, where the balance between structural and dynamical considerations for achieving signaling efficiency is, and the design principle that neurons optimize to preserve this balance has remained elusive. We took advantage of the theoretical prediction of theorem \ref{theo:theo3} discussed above and went looking to see if biological neurons displayed optimized signaling in the context predicted by the refraction ratio. One interpretation of our results is that the convoluted paths taken by axons reflect a design compensation by the neuron to slow down signaling latencies in order to optimize the refraction ratio. This was not a serendipitous discovery. We went looking to see if neurons actually optimize the refraction ratio because the mathematics pointed us in that direction.

\section*{Acknowledgments}
I am grateful to Prof. Fang Chung (Department of Mathematics, UC San Diego) for her input and feedback when this work was in its earliest stages, and to (the soon to be) Dr. Vivek George (Department of Bioengineering, UC San Diego) for the many fruitful discussions. I would also like to sincerely acknowledge the effort of the two anonymous reviewers. This was not an easy paper to review. Their very thoughtful feedback has strengthened the final work product. This work was supported by grants 63795EGII and N00014-15-1-2779 from the Army Research Office (ARO), United States Department of Defense, and in part from unrestricted funds to the Center for Engineered Natural Intelligence.

\end{document}